\documentclass[sigconf,authorversion,nonacm,screen]{acmart}
\pdfoutput=1

\usepackage{todonotes}
\usepackage[noline,linesnumbered]{algorithm2e}
\usepackage{physics}

\DeclareMathOperator*{\argmax}{argmax}
\newcommand{\procedurecall}[2]{\textsc{#1}(#2)}

\acmConference[KDD '22]{KDD 2022}{August 14--18, 2022}{Washington DC}
\acmBooktitle{KDD '22}
 \acmPrice{}
 \acmISBN{}

\begin{document}

\title{HyperLogLogLog: Cardinality Estimation With One Log More}

\author{Matti Karppa}
\email{mattk@itu.dk}
\affiliation{
  \institution{BARC}
  \institution{IT University of Copenhagen}
}
\author{Rasmus Pagh}
\email{pagh@di.ku.dk}
\affiliation{%
  \institution{BARC}
  \institution{University of Copenhagen}
}

\begin{abstract}
  We present HyperLogLogLog, a practical compression of the
  HyperLogLog sketch that compresses the sketch from $O(m\log\log n)$
  bits down to $m \log_2\log_2\log_2 m + O(m+\log\log n)$ bits
  for estimating the number of distinct elements~$n$ using
  $m$~registers. The algorithm works as a drop-in replacement that
  preserves all estimation properties of the HyperLogLog sketch, it is
  possible to convert back and forth between the compressed and
  uncompressed representations, and the compressed sketch maintains
  mergeability in the compressed domain. The compressed sketch can be
  updated in amortized constant time, assuming $n$ is sufficiently
  larger than $m$. We provide a C++ implementation of the sketch, and
  show by experimental evaluation against well-known implementations
  by Google and Apache that our implementation provides small sketches
  while maintaining competitive update and merge times. Concretely, we
  observed approximately a 40\% reduction in the sketch
  size. Furthermore, we obtain as a corollary a theoretical algorithm
  that compresses the sketch down to
  $m\log_2\log_2\log_2\log_2 m+O(m\log\log\log m/\log\log m+\log\log n)$ bits.
\end{abstract}



\maketitle

\section{Introduction}

Counting the number of distinct elements, or the \emph{cardinality},
of a data stream is a basic operation that has applications in network
traffic analysis~\cite{LiuCG:2016}, organization of large
datasets~\cite{HalevyKNOPRW:2016}, genome
analysis~\cite{BreitwieserBS:2018}, analysis of social
networks~\cite{BackstromBRUV:2012}, including large-scale industrial
applications such as estimating the number of distinct Google search
queries~\cite{HeuleNH:2013}. The na\"ive approach of storing all
unique elements in the data stream quickly becomes prohibitive as the
number of distinct elements in the data stream grows into the order of
billions, which calls for sketching approaches that maintain a sketch
of the cardinality using a limited number of bits.

Apart from the sheer number of bits, other design considerations for
the sketches include idempotence of updates, that is, that repeated
updates with the same element never change the sketch, control on the
estimation error, efficient updates, such that the sketch can be
updated even if the amount of computational resources is limited, and
mergeability of sketches, such that the data stream can be processed
by multiple computational units at once and the sketches can be
merged to produce output identical to as if a single machine had read
all the input. The HyperLogLog sketch by Flajolet, Fusy, Gandouet, and
Meunier~\cite{FlajoletFGM:2007} is a practical example of such a
sketch and has become a standard technique for cardinality
estimation. The HyperLogLog sketch makes use of $m$ counters, or
\emph{registers}, of $O(\log\log n)$ bits to produce an estimate of the
cardinality with a relative standard error of $1.04/\sqrt{m}$.

\paragraph{Our contribution}
We show that if we store offsets from a base
value in the registers and allow $O(m/\log m)$ registers to overflow,
storing the overflowing registers sparsely using $O(\log m)$ bits each, then,
with high probability,
$m\log_2\log_2\log_2 m+O(m+\log\log n)$ bits suffice
to represent the HyperLogLog sketch. The sketch is a
true compression of the HyperLogLog sketch, preserves all
estimation properties of HyperLogLog, and is amenable to the tricks
that can be used to boost the estimation accuracy of HyperLogLog
sketches. Furthermore, the sketch preserves mergeability in the
compressed domain without having to be decompressed for merging, and we
show that the update time is amortized constant for $n$ sufficiently
larger than $m$.
As a corollary, we obtain a sketch of
$m\log_2\log_2\log_2\log_2 m + O(m\log\log\log m /\log\log m+\log\log n)$ bits if a succinct
dictionary is used for the sparse representation.  

We also provide a C++ implementation of the sketch and show
empirically that it is applicable as a drop-in replacement for
HyperLogLog in practice and can be implemented with a competitive
performance over a wide parameter range.

\paragraph{Related work}
The HyperLogLog branch of research was initiated by Flajolet and
Martin in~\cite{FlajoletM:1985} where the idea of their algorithm was
to construct an $m\times w$ bit matrix using a random hash function to
divide the data stream into $m$ independend streams, corresponding to
the rows of the matrix, where the high bits of the hash value select
the substream, and the location of the least significant one-bit in
the low-bits of the hash value determines the column; the value thus
determined is set to 1. Each row then produces an estimate of the
cardinality of the set, and an arithmetic mean is used to lower the
variance of the estimate. In~\cite{DurandF:2003}, this was optimized
by only storing the maximum of the \emph{locations} of the first
one-bits, requiring $\log \log n$ bits, resulting in the LogLog
algorithm and its more practical variant, the
SuperLogLog. Eventually, this lead to
HyperLogLog~\cite{FlajoletFGM:2007}, the key change of which 
was the use of the harmonic mean instead of the arithmetic mean for lower variance.

The HyperLogLog sketch has seen a lot of practical applications and
improvements, including but not limited to the
HyperLogLog++~\cite{HeuleNH:2013} that includes practical engineering
work with numerous tricks for reducing memory usage and improving the
accuracy of the estimates, and the HLL-Tailcut~\cite{XiaoCZL:2020}
where the registers store the offsets with respect to the minimum
register value, and overflowing values are discarded.

Other approaches to counting distinct elements include Linear
Counting~\cite{WhangVT:1990}, and
MinCount~\cite{Bar-YossefJKST:2002,Giroire:2009} that was used by
Google before being replaced with
HyperLogLog++~\cite{HeuleNH:2013}. There exist a number of
\emph{non-mergeable} methods for cardinality estimation in recent
literature that offer potentially better space--accuracy tradeoffs
than their mergeable counterparts, and also methods to convert
mergeable sketches into more efficient sketches at the cost of
non-mergeability. See Pettie, Wang, and Yin~\cite{PettieWY:2021} for
an overview and analysis of such methods, along with their
\emph{Curtain} sketch.

The HyperLogLog sketch is known to be non-optimal in space
usage. Both the
Flajolet-Martin~\cite{FlajoletM:1985} and the HyperLogLog
sketch~\cite{FlajoletFGM:2007} are known to have an entropy of $O(m)$
bits~\cite{PettieWY:2021}, suggesting the application of entropy
compression methods to achieve the lower bound, although at the
expense of losing constant time updates. In~\cite{Lang:2017}, Lang
applies entropy compression to the original Flajolet-Martin
sketch~\cite{FlajoletM:1985} to get a practical sketch below the
entropy of the HyperLogLog sketch.  In~\cite{KaneNW:2010}, Kane,
Nelson, and Woodruff present a theoretical algorithm that achieves the
optimal $O(m)$ space bound with $O(1)$ updates. This was further
improved by B\l{}asiok~\cite{Blasiok:2020} by reducing the amount of
space required for parallel repetitions of the algorithm for reducing
the probability of failure. However, we are not aware of practical
implementations of these optimal algorithms, and regard them as mainly
of theoretical interest.\footnote{A particular difficulty in implementing these constructions is the use of \emph{Variable-Bit-Length Arrays} to store an array of variable-length counters compactly. The construction of Blandford and Blelloch~\cite{BlandfordB:2008} uses constant time per operation, but is complicated and a direct implementation appears to use space that is at least 3 times larger than the entropy of the counters.}

\paragraph{Organization}
The remainder of this paper is organized as follows. In
Section~\ref{sect:preliminaries}, we present the mathematical notation
and preliminaries used in the remainder of the paper. In
Section~\ref{sect:hyperlogloglog}, we recap HyperLogLog and present
our modifications that yield the HyperLogLogLog
algorithm, prove its space usage, 
and present changes that make the algorithm practical. In
Section~\ref{sect:implementation} we present an overview of our
implementation and engineering choices. Finally, in
Section~\ref{sect:experiments} we report on an empirical study where
we show that our algorithm is practical in terms of runtime and
provides space-efficient sketches that yield estimates equal to those
of HyperLogLog, and compare the running times and sketch sizes to
other well-known implementations.

\section{Preliminaries}
\label{sect:preliminaries}

We write $[n] = \{ 1, 2, \ldots, n \}$. We say that $w$ is the
\emph{word size}, meaning that it matches the number of bits output by
our hash functions and is somehow an efficient unit of data as
processed by an underlying computer architecture. When the actual
implementations are concerned, we are implicitly assuming $w=64$
unless noted otherwise. All logarithms are base 2 unless otherwise
stated, and $\ln$ is the natural logarithm.

We denote the universe over which we want to count the number of
distinct elements by $\mathcal U$. In our experimental work,
$\mathcal U$ is either a subset of integers or ASCII strings. We
assume we have access to a family of random hash functions
$h : \mathcal U \to [2^w]$ that map the elements of the universe to
$w$-bit integers uniformly at random.

When we say that $M\in U^m$ is an \emph{array}, we mean that $M$
consists of $m$ elements from the set $U$, and denote the lookup of
the $j$th element by $M[j]$ and the assignment of element $u$ to the
$j$th position by $M[j] \gets u$.

When we say that $S\subseteq K\times V$ is an \emph{associative
  array}, we mean that $S$ is a variable-sized data structure that
stores tuples $(k,v)\in K\times V$, or equivalently keys $k\in K$ that
map to values $v\in V$. We assume $S$ supports the operations of
membership query $k\in S$ to test whether the key $k$ has an
associated value in $S$, retrieval of values by key $S[k]$, assigning
the value $v$ to the key $k$ by $S[k] \gets v$, and removing a
key-value pair $\text{\texttt{del }}S[k]$. We tacitly assume that
there is only ever at most one value associated with any particular
key, and that the assignment of a pre-existing key with a new value
replaces the old value, so we treat $S$ as a map.

The function $\rho : [2^w]\to [w]$ is defined to be the 1-based index
of the first one-bit in the bit representation of the ($w$-bit)
integer. That is, for the bit sequence with $k-1$ initial zeros
$\rho(0^{k-1}1)=k$. We leave $\rho(0)$ undefined.

We say that random variables $X_1,X_2,\ldots,X_m$ are \emph{negatively
  dependent} if, for all $i\neq j$, they satisfy
$E[X_iX_j] \leq E[X_i]E[X_j]$. We say that a random event happens \emph{with high probability} if we can choose a constant $0<\beta<1$ such that the event happens with probability $1-\beta$. 
We use the
Iverson bracket notation to denote an indicator variable
$\llbracket \varphi \rrbracket$ that receives 1 if and only if the
expression $\varphi$ is true and 0 otherwise. We need the following form of the Chernoff bound.
\begin{lemma}[{Chernoff~\cite[Equation~1.7]{DubhashiP:2009}}]
  \label{lem:chernoff}
  Let $X=\sum_{j\in[m]} X_j$ be a sum of negatively dependent variables with all $X_j\in[0,1]$. Then, for all $\epsilon>0$,
  \[
    \Pr[X>(1+\epsilon)E[X]]\leq\exp\left(-\frac{\epsilon^2}{3}E[X]\right)\, .
  \]
\end{lemma}

\section{The HyperLogLogLog Algorithm}
\label{sect:hyperlogloglog}

\subsection{Problem statement}

We formally state the problem as follows. Given a sequence of elements
$(y_1,y_2,\ldots,y_s)$ from a universe $\mathcal U$, determine the
number of distinct elements $n=|\{ y_i : i\in [s] \}|$. 
Unless otherwise stated, we will assume that $n$ is the correct (but
unknown) cardinality of the datastream in question.

\subsection{HyperLogLog recap}

Since our algorithm is a modification on HyperLogLog, it is necessary
to understand how HyperLogLog works, so we start by restating the
HyperLogLog algorithm~\cite{FlajoletFGM:2007} using our notation. The
underlying idea of the algorithm is that, if we hash the elements
uniformly at random, we expect to see exactly one half of the elements
begin with a 1, exactly one fourth begin with the bitsequence 01, one
eighth begin with the bitsequence 001, and so on; so the largest
position~$r$ of the first one-bit witnessed over the sequence is a
(weak) indication that the stream has a cardinality of approximately
$2^r$.

Let us assume that we work over a universe $\mathcal U$. The
HyperLogLog data structure consists of an array $M\in[w]^m$ of $m$
elements of $\log w$ bits where $m$ is a power of two.\footnote{This
  requirement is not strictly necessary, but in practice the way the
  algorithm implemented, this is a reasonable
  assumption.} Furthermore, we fix two\footnote{In the original
  exposition~\cite{FlajoletFGM:2007}, only one 32-bit hash function
  was used and the most significant bits were used for register
  selection and the least significant bits as input for the function
  $\rho$; this can be seen as a special case of our treatment.} random
hash functions $h : \mathcal U \to [2^w]$ and $f : \mathcal U \to [m]$
by drawing them from the family of random hash functions. Initially we
set all elements of $M$ to be zeros. We say that the elements of $M$
are \emph{registers}. We make the observation that $w=\Omega(\log n)$
by the pigeon hole principle, as otherwise hash collisions necessarily
mask some of the distinct elements.

The data structure supports three operations: \procedurecall{update}{$y$} that
updates the data structure with respect to the element $y\in\mathcal
U$, \procedurecall{estimate}{} that returns the present cardinality
estimate, and \procedurecall{merge}{$M_1,$ $M_2$} that takes two sketches as
input and returns a third sketch that corresponds to the sketch that
would have been obtained if the elements that were used to update the
individual sketches had been directly applied on the output
sketch---assuming the two input sketches have been constructed using
the same hash functions.

Updates are performed by selecting the index of a register by
computing $j=f(y)$, and then update the value of the register to be
the maximum of its present value and the value $\rho(h(y))$. That is,
the invariant maintained by the algorithm is that the register $M[j]$
holds the maximum $\rho$-value over the elements assigned to the $j$th
substream by the hash function $f$. Since the hash values are $w$ bits
long, we have $0 \leq \rho(h(y)) < w$, so the values can be
represented with $\log w = \Theta(\log \log n)$ bits.

To construct an estimate from the register values, we compute the
bias-corrected, normalized harmonic mean of the estimates for the
substreams
\[
  E=\alpha_m m^2 \cdot \left( \sum_{j=1}^m 2^{-M[j]} \right)^{-1} \, ,
\]
where the bias-correction term $\alpha_m$ satisfies $\alpha_{16} =
0.673$, $\alpha_{32} = 0.697$, $\alpha_{64}=0.709$, and $\alpha_m =
0.7213/(1+1.079/m)$ for $m\geq 128$; see~\cite{FlajoletFGM:2007} for
the details. Finally, the algorithm includes further bias-correction
by applying Linear Counting for small cardinality ranges, and a
similar correction for large ranges.

Merging is very simple: simply construct a new array the elements of
which are the elementwise maxima of the input sketches. These
procedures are presented in pseudocode in
Algorithm~\ref{alg:hyperloglog}. The resulting estimate has low bias
(but is not unbiased), and it can be shown that the relative standard
error of the estimate is approximately
$1.04/\sqrt{m}$~\cite{FlajoletFGM:2007}.

\begin{algorithm}[t]
  \DontPrintSemicolon
  \SetKwProg{Proc}{Procedure}{}{{End Procedure}}
  \SetFuncSty{textsc}
  \SetKwFunction{update}{update}
  \SetKwFunction{estimate}{estimate}
  \SetKwFunction{merge}{merge}
  \caption{HyperLogLog.}
  \label{alg:hyperloglog}
  Let $m$ be a power of two. Initialize array $M\in[w]^m$ to be all zeros.
  Let $h : \mathcal U \to [2^w]$ and $f : \mathcal U \to [m]$ be fixed
  random hash function. \;
  \Proc{\update{$M,y$}}{
    $x \gets h(y)$\;
    $j \gets f(y)$\;
    $M[j] \gets \max \{ M[j], \rho(x) \}$\;
  }
  \Proc{\estimate{$M$}}{
    $E \gets \alpha_m m^2 \cdot \left( \sum_{j=1}^m 2^{-M[j]}
    \right)^{-1}$\;
    $V \gets |\{ j : M[j] = 0\}|$ \;
    \Return
    $\left\{
      \begin{array}{ll}
        m\ln\frac mV & \textrm{if } E \leq \frac{5}{2}m
                       \textrm{ and } V \neq 0 \\
        -2^{32} \ln\left(1-\frac{E}{2^{32}}\right) & \textrm{if } E >
                                                     \frac{2^{32}}{30} \\
        E & \textrm{otherwise}
      \end{array}\right.$\;
  }
  \Proc{\merge{$M_1,M_2$}}{
    Initialize $M\in[w]^m$\;
    \For{$j\in[m]$}{
      $M[j] \gets \max\{M_1[j],M_2[j]\}$\;
    }
    \Return $M$
  }
\end{algorithm}

\subsection{Asymptotic argument}
It is well-known\footnote{As mentioned, for example,
  in~\cite{PettieWY:2021}. Lang~\cite{Lang:2017} suggests that the
  constant is just below 2.832 which is supported by numerical
  experiments with Lemma~\ref{lem:prmjeqk}.} that $O(m)$ bits suffice
for an entropy-compressed version of HyperLogLog. However,
entropy-compression is expensive and does not allow for efficient
updates. We will show that by using $O(\log \log n)$ bits to store a
\emph{base} value $B$ and an array of offsets $M$ as a dense array,
plus a limited number of $(j,M[j])$ pairs as a \emph{sparse}
associative array, $m\log\log\log m$ + $O(m)$ bits suffice for the
sketch with high probability.

Suppose we feed $n$ distinct values to the
HyperLogLog sketch. Then each of the registers $M[j]$ can be treated as a
random variable, and the distribution of the register values is
determined by the following lemma.
\begin{lemma}[{\cite[Appendix~A]{XiaoCZL:2020}}]
  \label{lem:prmjeqk}
  After updating the HyperLogLog sketch of $m$ registers with $n$
  distinct values, the distribution of each $M[j]$ satisfies
  \[
      \Pr[M[j] = k] = \left\{
      \begin{array}{ll}
        \left(1 - \frac 1m\right)^n & \mathrm{if} \  k = 0\, , \\
        \left(1-\frac{1}{m2^k}\right)^n - \left(1-\frac{1}{m2^{k-1}}\right)^n &
                                                          \mathrm{otherwise}
                                                          \, .
      \end{array}
    \right.
  \]
\end{lemma}
\begin{corollary}
  \label{cor:prmjleqk}
  The distribution of $M[j]$ satisfies
  \[
    \Pr[M[j] \leq k] = \left(1-\frac{1}{m2^k}\right)^n
    \, .
  \]
\end{corollary}

Note that Lemma~\ref{lem:prmjeqk} only applies to any particular
register $j$, but not to all as a whole, since the registers are
negatively dependent. Throughout this section, we are going to treat
the register values as real-valued, continuous random variables, as
the map $k\mapsto 1-\frac{1}{m2^k}$ is monotone, only rounding to
integers at the very end. We also need the following lemma.
\begin{lemma}
  \label{lem:registertails}
  Let $B=\log\frac{n}{m}$. Then, for $\Delta\in(0,B)$ and each
  $j\in[m]$,
  \begin{itemize}
  \item $\Pr[M[j] < B-\Delta] < \exp\left(-2^\Delta\right) <
    2^{-\Delta}$, and
  \item $\Pr[M[j] > B+\Delta]$ $< 2^{-\Delta}$.
  \end{itemize}
\end{lemma}
\begin{proof}
  It is well-known that $(1-x)^n < \exp(-nx)$ for all $0<x<1$. 
  By Corollary~\ref{cor:prmjleqk} and treating $k$ as a real variable,
  \[
    \begin{split}
      \Pr[M[j] < B-\Delta] & = \left(1-\frac{1}{m2^{B-\Delta}}\right)^n \\
      & < \exp\left(-\frac{n}{m2^{B-\Delta}}\right) = \exp\left(-2^\Delta\right) < 2^{-\Delta} \, .
    \end{split}
  \]
  Likewise, by complement from Corollary~\ref{cor:prmjleqk} and by approximating $(1-x)^n<1-nx$ for $0<x<1$,
  \[
    \begin{split}
      \Pr[M[j] > B+\Delta] & = 1 - \left(1-\frac{1}{m2^{B+\Delta}}\right)^n \\
      & < \frac{n}{m2^{B+\Delta}} = 2^{-\Delta}\,.
    \end{split}
  \]
\end{proof}

We will start by showing that, for sufficiently large $n$ and $m$, the
sketch can be represented with $O(m\log\log m+\log\log n)$ bits, with high probability.

\begin{theorem}
  \label{thm:main:loglog}
  For $\beta\in(0,1)$ and $n>2m^2/\beta$, all register values can be represeted as offsets from the base value $B=\log\frac nm$ using at most $\lceil \log\log \frac{2m}{\beta}\rceil + 1 = O(\log\log m)$ bits, with probability at least $1-\beta$.
\end{theorem}
\begin{proof}
  Fix the constant $0<\beta<1$. We will set
  $\Delta=\log\frac{2m}{\beta}$. By our condition on $n$, $\Delta<B$,
  so we apply Lemma~\ref{lem:registertails}, and get, for each
  $j\in[m]$,
  $\Pr[M[j]\not\in(B-\Delta,B+\Delta)] < 2\cdot 2^{-\Delta} =
  \frac{\beta}{m}$. By the union bound over the $m$ registers, we get
  that all registers are within this interval with probability at
  least $1-\beta$. Finally, by our choice of $\Delta$, we can encode
  an integer in the desired range of $(B-\Delta,B+\Delta)$ by using at
  most $\log(2\Delta) = \lceil \log\log\frac{2m}{\beta}\rceil + 1$
  bits.
\end{proof}

Theorem~\ref{thm:main:loglog} is similar to the HLL-Tailcut approach
of~\cite{XiaoCZL:2020}, and could indeed be used to show that the tailcut approach requires
asymptotically only few bits.  However, we can do better than this and
use only $\log\log\log m+O(1)$ bits per register by showing that
$O(m)$ bits suffice for representing the overflowing registers
sparsely as $(j,M[j])$ pairs.

\begin{theorem}
  \label{thm:main:logloglog}
  Suppose $n>4m\log m$. Then at least $m-\frac{m}{\log m}$ register values can be represented as offsets from the base value $B=\log\frac{n}{m}$ using at most $\lceil \log(2+\log\log m)\rceil+1 = \log\log\log m+O(1)$ bits, with probability at least $1-\exp(-m/(6\log m))$. 
\end{theorem}
\begin{proof}
  We will set $\Delta = 2+\log\log m$. By our condition on $n$, $\Delta<B$, so we can apply Lemma~\ref{lem:registertails}. For each $j\in[m]$, we define an indicator variable $X_j = \llbracket M[j]\not\in(B-\Delta,B+\Delta)\rrbracket$. We note that $X_j\sim\mathrm{Bernoulli}(p)$. We get from Lemma~\ref{lem:registertails} that $p$ satisfies 
  \[
    p = \Pr[X_j = 1] < 2\cdot 2^{-\Delta} = \frac{1}{2\log m} \, .
  \]
  We then define the sum variable $X=\sum_{j\in[m]}X_j$ and observe that, by linearity of expectation, $E[X] < m/(2\log m)$. Since the variables $X_j$ are negatively dependent, we can apply the Chernoff bound of Lemma~\ref{lem:chernoff} to bound the number of registers whose values do not fall in our desired interval
  \[
    \Pr[X>\frac{m}{\log m}] \leq \exp\left(-\frac{m}{6\log m}\right) \, .
  \]
  Finally, we note that by our choice of $\Delta$, any integer in $(B-\Delta,B+\Delta)$ can be encoded using $\lceil \log(2\Delta)\rceil = \lceil\log(2+\log\log m)\rceil + 1 = \log\log\log m+O(1)$ bits, and the sparse representation requires $O\left(\frac{m}{\log m}\right)\cdot O(\log{m}) = O(m)$ bits by the fact that the register indices are in $[m]$ and Theorem~\ref{thm:main:loglog}.
\end{proof}

As a corollary of Theorem~\ref{thm:main:logloglog}, we obtain an
algorithm that we might call HyperLogLogLogLog that uses
$m\log\log\log\log m+O(m\log\log\log m / \log\log m+\log\log n)$ bits for the
sketch if we use a succinct dictionary for the sparse
representation. We believe this algorithm is only a theoretical
curiosity, but present it for completeness.

\begin{corollary}
  \label{cor:main:loglogloglog}
  Suppose $n>4m\log\log m$. Then at least $m-\frac{m}{\log\log m}$ register values can be represented as offsets from the base value $B=\log\frac{n}{m}$ using at most
  \[
    \lceil \log(2+\log\log\log m)\rceil + 1 = \log\log\log\log m + O(1)
  \]
  bits, with probability at least $1-\exp(-m/(6\log\log m))$. Using a
  succinct dictionary for the sparse representation, the sketch has a
  total size of $m\log\log\log\log m+O(m\log\log\log m/\log\log
  m+\log\log n)$ bits.
\end{corollary}
\begin{proof}
The proof is otherwise the same as in Theorem~\ref{thm:main:logloglog}
except we use $\Delta = 2+\log\log\log m$. For the sparse
representation, we note that it is known that succinct dictionaries
(see, for example, \cite{Pagh:2001}) require $O(s\log\frac{m}{s})$
bits where $s$ is the number of distinct elements in the dictionary
and $m$ is the size of the universe. In particular, the same $m$ as in
our case works since the elements are pairs in $[m]\times [\log m]$ by
Theorem~\ref{thm:main:loglog}, with high probability, so the number of
bits required in the standard representation for any element in the
universe is $O(\log m)$. Applying the bound on the succinct dictionary
size, together with $s=m/\log\log m$, by our choice of $\Delta$, and
the number of bits required for the base value,
we get the bound of $m\log\log\log\log m + O(m\log\log\log m /
\log\log m+\log\log n)$ bits on the total size of the sketch.
\end{proof}

Finally, we note that storing the base value $B$ requires $O(\log\log
n)$ bits which we can afford since we are going to need $O(\log n)$
auxiliary space for storing the computed hash values anyway. This is
also reflected in the more precise bound of~\cite{KaneNW:2010}.

\subsection{Practical algorithm}

The practical applicability of Theorems~\ref{thm:main:logloglog} is
somewhat limited as the $\log\log\log m$ is an extremely slowly
growing function.  Indeed, it is commonly believed that there are
roughly $10^{80}$ atoms in the observable universe, and
$\log\log\log 10^{80} < 3.01$ which is hardly a large
number. Furthermore, the proofs only address the size of the resulting
sketch and make use of $n$ which we cannot do in practice.

However, numerical simulations reveal that using 3 bits per register
tends to yield small sketches for practical values of $m$ and $n$, and
so the intuition of the sparse/dense split seems well applicable in
practice. We will introduce a separate parameter $\kappa$ that
controls the number of bits per dense register. For our
implementation, we fix $\kappa=3$ (along with $w=64$). This is quite close to the entropy bound of $2.832$~\cite{Lang:2017}, so we cannot even hope to do much better.

Let $M'\in[w]^m$ be the registers of the uncompressed HyperLogLog sketch. 
We denote the
fixed-size dense array of offsets by $M\in[\kappa]^m$, and the base
value $B\in[w]$. The dense part corresponds to the ``fat region'' of
the distribution of values around the expectation, and the registers
in the fat region satisfy
\begin{equation}
  \label{eq:denseregister}
  B\leq M'[j] < B+2^\kappa \, ,
\end{equation}
and we store $M[j] = M'[j] - B$ using $\kappa$ bits.  The sparse part
corresponds to registers $j$ that do \emph{not} satisfy
Equation~\eqref{eq:denseregister}, and they are stored in an
associative array $S\subseteq [m]\times[w]$ whose elements are $(j,r)$
pairs where $r=M'[j]$, or equivalently, $S$ maps the sparsely
represented $j$ to their corresponding register values. Each element
of $S$ takes $\log m + \log w$ bits.

Thus, the HyperLogLogLog data structure is a 3-tuple $(S,M,B)$, and
initially $S=\emptyset$, $M$ is initialized to all-zeros, and
$B=0$. The size of the data structure is variable, and is determined
unambiguously by the number of elements stored in $S$, or
equivalently, the choice of the base value $B$. The total size of the
data structure is
\[
  m\kappa+|S|\cdot (\log m+\log w) \, .
\]
If $B$ is
chosen so as to minimize the number of bits required, we say that
\emph{the sparse-dense split is optimal}.

After each update, if any of the register values was changed, we
maintain as an invariant that the split is optimal by running a
subroutine which we call \procedurecall{compress}{}. The subroutine determines
the optimal $B'$, and if this differs from the present $B$, performs a
\procedurecall{rebase}{} operation whereby registers are reassigned into sparse
or dense representations, depending on which side of the split determined by $B'$ they
fall.

The \procedurecall{update}{} procedure is shown as pseudocode in
Algorithm~\ref{alg:hyperlogloglog}. The full runtimes of each
operation depends on the actual choice of data structures, but it is
immediate that the compression routine is rather expensive: we need to
try up to $w$ different new base values, and evaluating the size of
the resulting data structure requires $\Omega(m)$ operations for each
different proposed new base value, so the compression routine requires
at least $\Omega(mw)$ operations. However, it is easy to see that if
the $n$ is sufficiently larger than $m$,
then only very few
elements actually trigger an update; thus, as the next lemma
shows, the update times are actually amortized constant over a
sufficiently large $n$. We note that Lemma~\ref{lem:amortizedconstant}
is not tight, as it makes some rather crude and pessimistic
approximations.

\begin{lemma}
  \label{lem:amortizedconstant}
  For sufficiently large $n$ satisfying $n/(\log n)^2>m^2$,
  the updates are amortized constant,
  with high probability.
\end{lemma}
\begin{proof}
  Fix a constant $\gamma \geq 1$.
  From Corollary~\ref{cor:prmjleqk} and by approximating $(1-x)^n \geq 1-nx$, we get for each $j$ that
  \[
    \Pr[M[j] > k] \leq \frac{n}{m2^k} \, .
  \]
  Let $k$ be the largest value that we expect to see. Setting
  $k=\log(\gamma n)$, we get by the union bound that all registers are
  below this value at probability $1-1/\gamma$.

  Let us then consider the total amount of work for the
  compressions. The compression is triggered at most $km$ times. Each
  time the compression routine is triggered, the amount of work
  required is $O(km)$ since we need to try at most $k$ base values and
  perform $O(m)$ work for each base value candidate. The total amount of compression work is thus $O(m^2k^2) = O(m^2(\log(\gamma n))^2)$.
  In addition to compressions, we also need $O(1)$ work for every update. Average work for $n$ distinct elements is thus$\frac{O(n)+O(m^2(\log n)^2)}{n} = O(1)$
  for sufficiently large $n$ by our assumption that $n/(\log n)^2>m^2$, thus an
  amortized constant.
\end{proof}

\begin{algorithm}[t]
  \DontPrintSemicolon
  \SetKwProg{Proc}{Procedure}{}{{End Procedure}}
  \SetKwFunction{get}{get}
  \SetKwFunction{update}{update}
  \SetKwFunction{compress}{compress}
  \SetKwFunction{rebase}{rebase}
  \SetKwFunction{countsparse}{countsparse}
  \SetFuncSty{textsc}
  \caption{HyperLogLogLog update procedure.}
  \label{alg:hyperlogloglog}
  Let $m$ be a power of two. Initialize array $M\in[\kappa]^m$ to be
  all zeros, the associative array $S\subseteq [m]\times[w]$ to be
  $S\gets \emptyset$, and the base $B\gets 0$.
  Let $h : \mathcal U \to [2^w]$ and $f : \mathcal U \to [m]$ be fixed
  random hash function. \;
  \Proc{\get{$S,M,B,j$}}{
    \Return $\left\{\begin{array}{ll}
                      S[j] & \textrm{if } j\in S \\
                      B+M[j] & \textrm{otherwise}
                    \end{array}\right.$\;
  }
  \Proc{\update{$S,M,B,y$}}{
    $j \gets f(y)$\;
    $r \gets \rho(h(y))$\;
    $r_0 \gets$\get{$S$,$M$,$B$,$j$}\;
    \If{$r>r_0$}{
      \uIf{$B\leq r < B+2^\kappa$}{
        $M[j] \gets r-B$\;
        \If{$j\in S$}{
          $\texttt{del }S[k]$\;
        }
      }
      \Else{
        $S[k]\gets r$\;
      }
      \compress{$S,M,B$}\;
    }
  }
  \Proc{\compress{$S,M,B$}}{
    $B'\gets \displaystyle\argmax_{B''\in[w]} |\{j : B'' \leq
    \get{S,M,B,j} < B'' + 2^\kappa\}|$\;
    \If{$B'\neq B$}{
      \rebase{$S,M,B,B'$}\;
    }
  }
  \Proc{\rebase{$S,M,B,B'$}}{
    \For{$j\in[m]$}{
      $r\gets \get{S,M,B,j}$\;
      \uIf{$B'\leq r<B'+2^\kappa$}{
        $M[j]\gets r-B'$\;
        \If{$j\in S$}{
          $\texttt{del }S[j]$\;
        }
      }
      \Else{
        $S[j]\gets r$\;
      }
    }
    $B\gets B'$\;
  }
\end{algorithm}

The other supported operations are the regular
HyperLogLog operations \procedurecall{estimate}{}, and
\procedurecall{merge}{}. Estimation is performed exactly as in the
case of HyperLogLog, with the exception that instead of array access,
we need to use the auxiliary \procedurecall{get}{} function to access
the elements. The \procedurecall{merge}{} operation is performed in
the compressed domain, so at no point is there need to perform a full
decompression into the regular HyperLogLog representation. The
\procedurecall{merge}{} operation is presented as pseudocode in 
Algorithm~\ref{alg:hyperlogloglogmerge}.

\begin{algorithm}[t]
  \DontPrintSemicolon
  \SetKwProg{Proc}{Procedure}{}{{End Procedure}}
  \SetKwFunction{get}{get}
  \SetKwFunction{update}{update}
  \SetKwFunction{compress}{compress}
  \SetKwFunction{rebase}{rebase}
  \SetKwFunction{countsparse}{countsparse}
  \SetKwFunction{merge}{merge}
  \SetFuncSty{textsc}
  \caption{HyperLogLogLog merge procedure.}
  \label{alg:hyperlogloglogmerge}
  \Proc{\merge{$S_1,M_1,B_1,S_2,M_2,B_2$}}{
    Initialize $S\gets \emptyset$, $M\in[2^\kappa]^m$ as zeros, and $B\gets\max\{B_1,B_2\}$\;
    \For{$j\in[m]$}{
      $r_1\gets$\get{$S_1,M_1,B_1$}\;
      $r_2 \gets$\get{$S_2,M_2,B_2$}\;
      $r\gets\max\{r_1,r_2\}$\;
      \uIf{$B\leq r<B+2^\kappa$}{
        $M[j] \gets r-B$\;
      }
      \Else{
        $S[j]\gets r$\;
      }
    }
    \compress{$S,M,B$}\;
    \Return $(S,M,B)$\;
  }
\end{algorithm}

\subsection{Performance improvements and heuristics}

There are several tricks that can be done to improve the practical
performance of the HyperLogLogLog algorithm. Some of these tricks have
effects on theoretical guarantees, some don't, and others are simply a
matter of implementation choices.

Importantly, choosing suitable data structures enables cache-friendly
linear access through the entire HyperLogLogLog structure, and we need
not actually use a potentially expensive \procedurecall{get}{}
function in the \procedurecall{compress}{}, \procedurecall{rebase}{},
or \procedurecall{merge}{} subroutines.

Furthermore, the updates are determined by pairs of random variables
$(j,r)$ where $j$ is distributed uniformly over $[m]$ and $r$ follows
a geometric distribution. This means that in most cases, \emph{no
  update} takes place after seeing sufficiently many elements, as the
vast majority of the $r$ values encountered are very small. If we use
$\log w$ bits to store the minimum register value, we can terminate
the \procedurecall{update}{} process early without any effect on
theoretical guarantees, without having to even look up the actual
register value.

Also, it is obvious that any reasonable base value $B$ should be equal
to an actual register value, so we need not try all $w$ options; also,
trying candidate base values in an ascending order enables us to
maintain a lower bound on the number of sparse elements, which leads
to a possible early termination of the \procedurecall{compress}{}
routine without loss of theoretical guarantees. We include these
optimizations in our implementation of HyperLogLogLog.

Considering the behavior of the random $(j,r)$ pairs, we see that, in
general, the behavior is rather benign,\footnote{This is assuming that
  we are \emph{not} dealing with an adversary who has access to the
  hash function or can affect the coin tosses when selecting the hash
  function; in this setting, a worst-case input can cause the
  algorithm to fail catastrophically. This is, however, a well-known
  property inherent to HyperLogLog itself, as the algorithm is not
  robust against adversarial input~\cite{PatersonR:2021}.} and trying
all possible base values in the \procedurecall{compress}{} routine is
mostly useless, as even though it is possible to construct nasty
corner cases that require unexpected rebasing operations to maintain
optimality, these seldom occur in practice with random
data. Therefore, we implement the following changes to a variant of
the algorithm we call HyperLogLogLog$^*$: we only call
\procedurecall{compress}{} if the size of the data structure needs to
be increased (that is, a new element needs to be added into $S$), and
we only try the next possible register value that is larger than $B$,
and omit all other choices. Experiments show that these heuristics
have little effect on the actual size of the resulting data structure,
but they improve the actual runtimes.

We also provide as a baseline a variant which we call HyperLogLogLogB
that fixes $B$ to be the minimum register value, and maintains a
counter that records the number of minimum-valued
registers. Functionally, this variant behaves like
HLL-Tailcut~\cite{XiaoCZL:2020}, except instead of allowing for error
when the register values overflow, it stores them sparsely. While
this is the fastest variant, as the compression becomes nearly
trivial, experiments show that the sketches are noticeably larger.

\section{Implementation}
\label{sect:implementation}

\subsection{Overview}

We provide a C++ implementation of the HyperLogLogLog,
Hyper\-Log\-Log\-Log$^*$, and HyperLogLogLogB algorithms, along with a
comparable implementation of the vanilla HyperLogLog, and an
entropy-compressed version of the HyperLogLog, compressed using the
Zstandard library.\footnote{\url{https://facebook.github.io/zstd/}} We
also provide our full experimental pipeline that can be built as a
Docker container for reproducibility. The code is available
online\footnote{\url{https://github.com/mkarppa/hyperlogloglog}}
under the MIT license. 
Functionally, our
implementation is a drop-in replacement for HyperLogLog and enables
conversion between the compressed and uncompressed representations,
whilst producing the exact same estimates.

The guiding principles for constructing the implementation have been
to minimize the memory usage and enable cache-friendly linear access
through the data structure, at the cost of potentially losing
optimization opportunities that would require auxiliary space. Our
implementation assumes a 64-bit environment, which is reflected in
the design choices for the data structures and hash functions.

\subsection{Data structures}

We implement both the array $M$ and the associative array $S$ as
bit-packed arrays of 64-bit words. An array of $n$ elements of size
$s$ bits will occupy $\lceil\frac{ns}{64}\rceil\cdot 64$ bits of
memory; importantly, we allow elements to cross word boundaries. For
the HyperLogLog implementation, the array $M$ consists of 6-bit
elements, so the array size is $\lceil \frac{6m}{64} \rceil \cdot 64$
bits, and for HyperLogLogLog, we use 3 bits per element, so the array
size is $\lceil \frac{3m}{64} \rceil \cdot 64$ bits.

We implement the associative array by concatenating the key-value
pairs into integers where the more significant bits are occupied by
the key and the less significant bits by the value. We maintain the
bit-packed array sorted in ascending order by applying insertion sort
after each insertion. We perform lookups by binary search in
$O(\log|S|)$ time. For the HyperLogLogLog associative array $S$, the
elements are of size $\log m + 6$, since we use 64-bit hash
functions. Maintaining the associative array $S$ sorted enables linear iteration over the entire HyperLogLogLog data structure.

\subsection{Hash functions}

Our implementation accepts unsigned 64-bit integers and 8-bit byte
strings as input. We also accept raw $(j,r)\in[m]\times[w]$ integer
pairs as input for evaluating runtimes without hash function
evaluations.

For the hash function $h:\mathcal U \to [2^w]$, we use Google's
Farmhash.\footnote{https://github.com/google/farmhash} Specifically
we use the function \texttt{Fingerprint} for 64-bit integers, and
\texttt{Hash64} for strings.

We derive the hash function $f : \mathcal U \to [m]$ from $h$ by
applying Fibonacci Hashing~\cite[Section~6.4]{Knuth:1998} on the hash
value $h(y)$, that is, we multiply the hash value by
$9e3779b97f4a7c15_{16}$ modulo $2^{64}$ and take the $\log m$ most
significant bits of the result.

\subsection{Entropy compression}

We have implemented an entropy-compressed version of HyperLogLog as a
baseline. The implementation behaves like vanilla HyperLogLog, but
the array is compressed with Zstandard library after each update. The
number of decompressions is limited by maintaining the information
about the minimum value in a separate variable, much like in the case
of our HyperLogLogLog implementation, and decompression is only
triggered if the encountered $\rho$-value is at least as large as the
minimum value in the registers. The idea of using entropy compression
is an obvious alternative, and we use Zstandard with compression level
set to 1, as suggested in~\cite{Lang:2017}.

\section{Experiments}
\label{sect:experiments}

\begin{table}[t]
  \begin{center}
    \caption{Implementations used in experiments.}
    \label{tbl:implementations}
    \begin{tabular}{@{}l@{\,\,}p{10em}@{\,\,}l@{\,\,}p{8em}@{}}
      \toprule
      Name & Description & Language & Notes \\
      \midrule
      \texttt{HLL}   & Our HyperLogLog & C++ & 6 bits per element\\
      \texttt{HLLL}  & Our HyperLogLogLog & C++ & Sorted array\\
      \texttt{HLLL*} & Our HyperLogLogLog$^*$ & C++ & Includes
                                                      heuristics\\
      \texttt{HLLLB} & Our HyperLogLogLogB & C++ & $B$ is minimum \\
      \texttt{HLLZ} & Zstd-compressed \texttt{HLL} & C++, C & Zstd is
                                                                 in C \\
      \texttt{AHLL8} & Apache HyperLogLog & C++ & 8 bits
                                                  per
                                                  element
      \\
      \texttt{AHLL6} & Apache HyperLogLog & C++ & 6 bits
                                                  per
                                                  element
      \\
      \texttt{AHLL4} & Apache HyperLogLog & C++ & 4 bits
                                                  per
                                                  element
      \\
      \texttt{ACPC} & Apache CPC & C++ & Compressed Probability Counting
      \\
      \texttt{Zeta} & Google ZetaSketch & Java & HyperLogLog++\\
      \bottomrule
    \end{tabular}
  \end{center}
\end{table}
\subsection{Experimental setup}

We evaluate a number of implementations by varying the number of
registers~$m$ as powers of two from $2^4$ to $2^{18}$, and generate random
input of~$n$ elements for $n=2^{\frac{i}{2}}$ for $i=\{4,5,\ldots,60\}$,
that is, $n$ is the power of a square root of 2 from $2^4$ to
$2^{30}$, rounded to the nearest integer.
The data consists of random unsigned 64-bit integers,
8-character-long alphanumeric random ASCII strings, and $(j,r)$-pairs
where $j$ has been drawn from $[m]$ uniformly at random and
$r\sim\mathrm{Geom}(0.5)$. The $(j,r)$ input is only supported by our
own implementations. We also evaluate our hash functions separately,
without performing a full update sequence.

We compare our implementations against Apache
DataSketches\footnote{\url{https://datasketches.apache.org/}}
implementation of HyperLogLog, Apache DataSketches implementation of
Compressed Probability Counting (CPC)~\cite{Lang:2017}, and Google's
ZetaSketch\footnote{\url{https://github.com/google/zetasketch}}
implementation of HyperLogLog++~\cite{HeuleNH:2013}. The
implementations are listed in Table~\ref{tbl:implementations}.

All experiments
are run in a Docker container that is managed by a Python
script. Data is generated by a separate C++ program, and the data is
then passed to a wrapper program that stores the data in RAM, so as to
avoid I/O issues.

There are two types of experiments: update experiments and merge
experiments. In the update case, the wrapper constructs a sketch by
feeding the data one element at a time. Once the sketch has been
constructed, the wrapper reports the time it took to construct the
sketch, the cardinality estimate computed from the sketch, and bit
size of the sketch.
For the merge experiment, we first divide the
data into two
equal-sized chunks and construct one sketch for each chunk.
We then report the time it takes to
construct the merged sketch from the two sketches. In all cases, 10
independent repetitions were performed with different data, but the
same data supplied for each implementation at a given $m$, $n$, and
repetition.

The reported bit sizes for our implementations are $6m$ for
\texttt{HLL}, 8 times the number of bytes occupied by the Zstd
compressed register array for \texttt{HLLZ}, and
$|S|\cdot(\log m+6)+3m$ for \texttt{HLLL}, \texttt{HLLL*}, and
\texttt{HLLLB}. For \texttt{Zeta}, following the documentation of the
library, we report $8m$. For Apache sketches, we use the bit size of the
output of the serialization function to upper bound the size of the sketch.

The experiments were run on a computer running two Intel Xeon
E5-2690v4 CPUs at 2.60~GHz with a total of 28 cores, and a total of
512~GiB of RAM, running Ubuntu 18.04.5~LTS. The experiments were run
using CPython~3.8.10 and Docker~20.10.7. The Docker environment ran
Ubuntu 20.04.3~LTS, and the C++ code was compiled with GCC
10.3.0. ZetaSketch was compiled with Gradle 7.2 and
OpenJDK~11.0.13. We ran up to 14 experiments in parallel.

\subsection{Results}
\begin{figure}[t]
  \begin{center}
    \includegraphics[width=\linewidth]{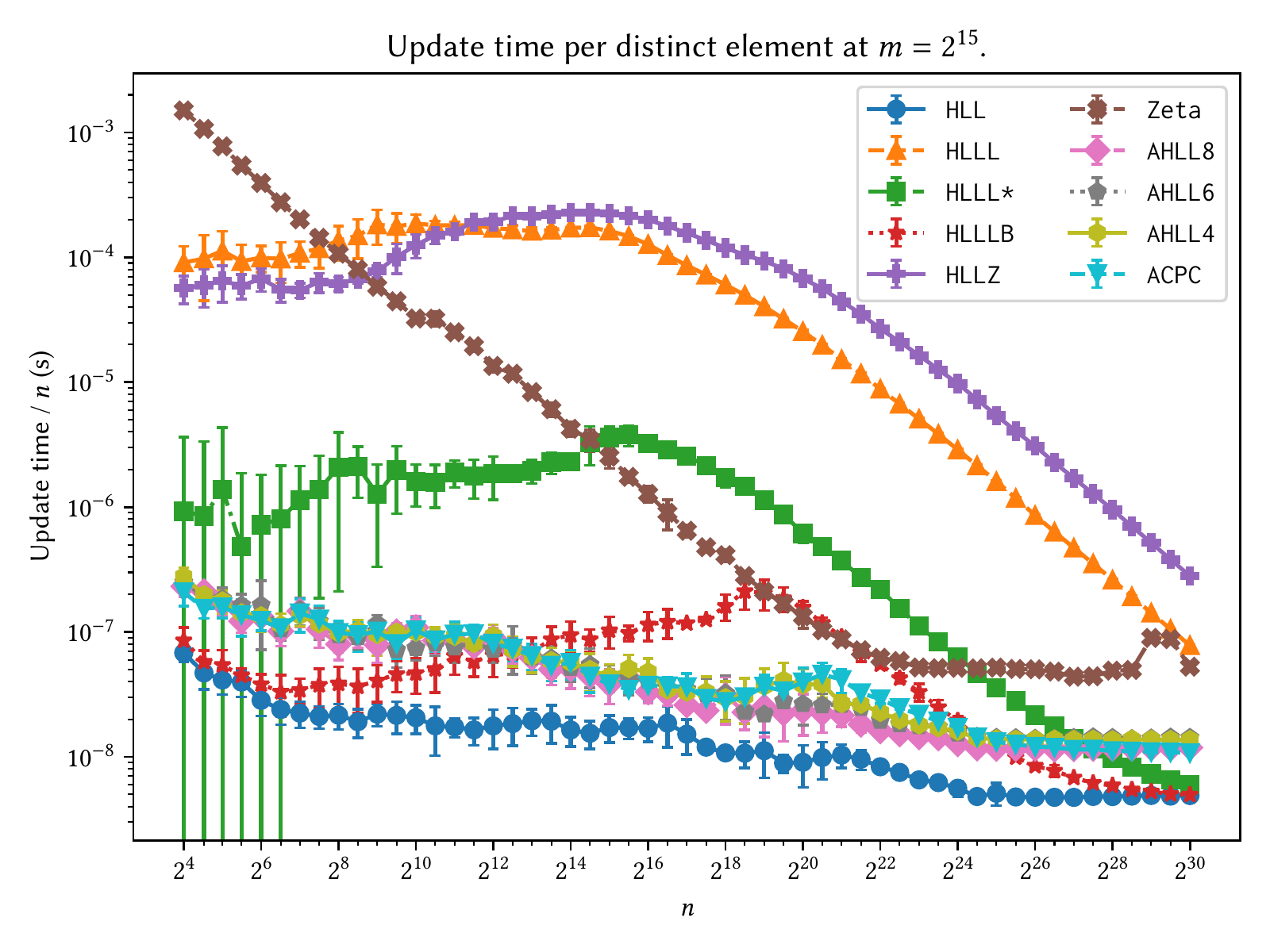}
    \caption{Scaling of update times of different algorithms at
      $m=2^{15}$.
      The times are given in seconds / distinct
      element. Random unsigned 64-bit integers were used as input.}
    \label{fig:scalingatm32768}
  \end{center}
\end{figure}

The scaling of update times is reported in
Figure~\ref{fig:scalingatm32768} at fixed $m=2^{15}$ registers. We
report the average time per distinct element for constructing the sketch. This
shows the general behavior of the algorithm: \texttt{HLL}
is clearly the fastest, but \texttt{HLLL*} catches up as the
number of distinct elements grows. For another view,
Figure~\ref{fig:totaltime1073741824} records the total time for
constructing a sketch with $n=2^{30}$ distinct elements, as a
function of $m$. \texttt{HLLL} and \texttt{HLLL*} remain
competitive with vanilla HyperLogLog until $m=2^{12}$ and $m=2^{15}$,
respectively, and then the total time starts to grow.

\begin{figure}[t]
  \begin{center}
    \includegraphics[width=\linewidth]{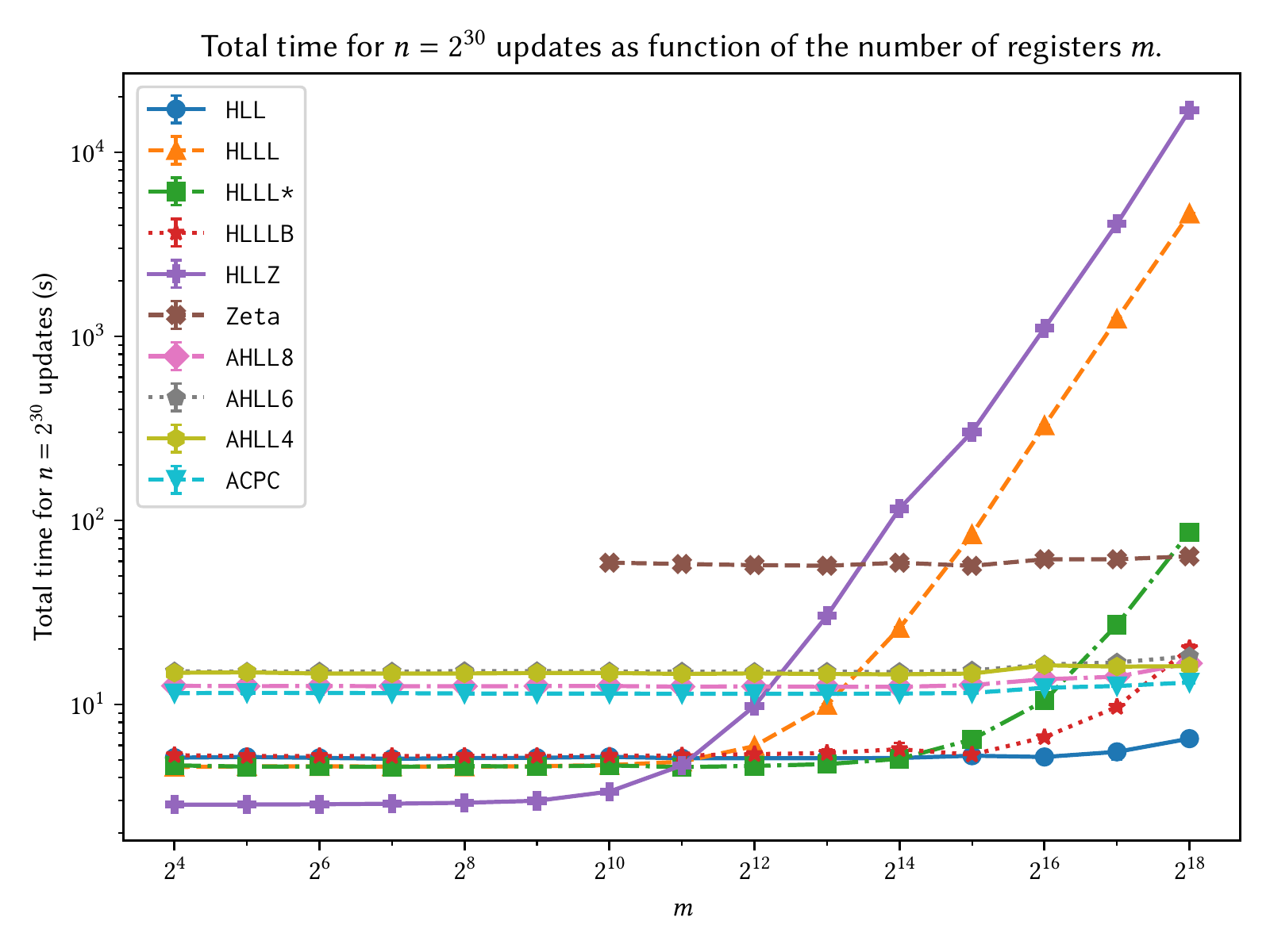}
    \caption{Total times for $n=2^{30}$ updates 
      as a function of $m$ in seconds. Random
      unsigned 64-bit integers are used as input.}
    \label{fig:totaltime1073741824}
  \end{center}
\end{figure}

Figure~\ref{fig:loadfactor} shows yet another view into this behavior
in terms of \emph{load factor} $n/m$. The figure shows the update time
per distinct element with \texttt{HLLL} as a function of the load
factor. Input data consists of $(j,r)$ pairs, so no hashing is
performed. Furthermore, the datapoints indicate the number of
registers used. We see that when the number of distinct elements is
sufficiently large with respect to the number of registers,
the algorithm achieves amortized constant behavior. Indeed, this is to be expected, as once the registers are filled
by distinct values, very few updates take place, as per
Lemma~\ref{lem:amortizedconstant}. For \texttt{HLLL},
in this particular case,
this regime is reached when the load factor is approximately
$2^{17}$; this is only achieved with $m\leq 2^{13}$ due to fixed
$n$. \texttt{HLLL*} reaches the same effect at a lower load factor,
but the results are more noisy, owing to the heuristics, but we have
omitted a separate figure for lack of space.

Although not shown here for lack of space, in particular when the
stream consists of strings, hashing can make a large portion of the
time. In fact, the vanilla HyperLogLog is so efficient in its updates,
that almost 100\% of time is spent on hashing when performing the
updates with strings as input. This suggests that slower updates are
not necessarily a problem, since there are other bottlenecks that may
be unavoidable. In particular, I/O operations can be orders of
magnitude slower, which we have deliberately avoided in our
experiments by performing everything in RAM.

\begin{figure}[t]
  \begin{center}
    \includegraphics[width=\linewidth]{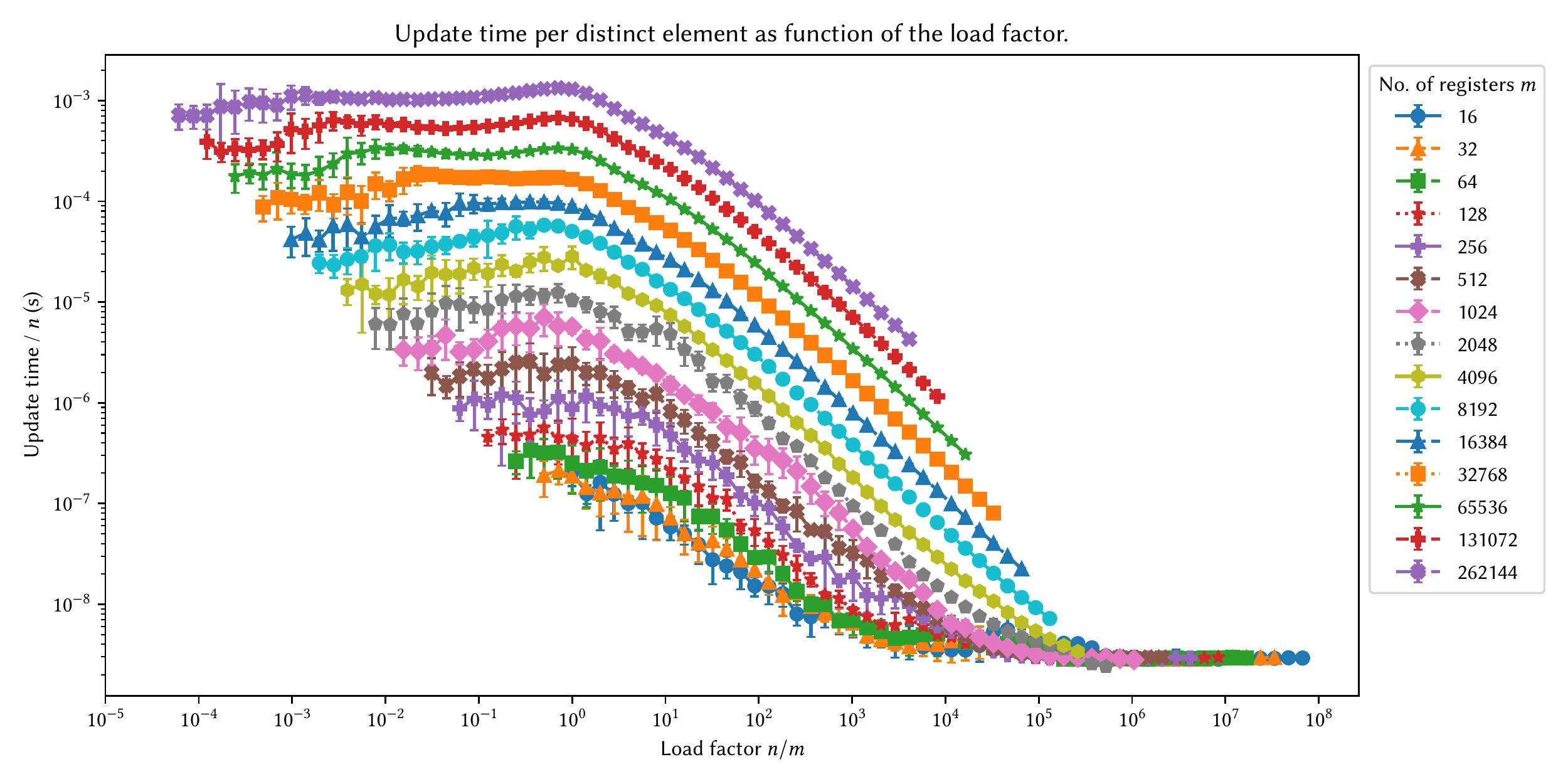}
    \caption{Update time of \texttt{HLLL} per distinct element as a
      function of the load factor $n/m$ with different $m$
      highlighted.
      The times are given in seconds per
      element. Random $(j,r)$ integer pairs are used as input, so no
      hashing is performed.}
    \label{fig:loadfactor}
  \end{center}
\end{figure}

Figure~\ref{fig:merge} shows the time required for merging of two
sketches that have been constructed using $n=2^{30}$ distinct
elements, split equally in half among the sketches, as a function of
the number of registers $m$. We see that,
for \texttt{HLLL},
merging becomes more
costly the larger the sketch size is, largely due to the application
of the full compression routine after the merge. The same routine is
applied also for \texttt{HLLL*}, so there is no discernible difference
in runtime from the \texttt{HLLL}. Despite being slower, we are still
talking about less than one third of a second for merging two
$m=2^{18}$ sketches. Perhaps surprisingly, \texttt{HLLLB} performs
poorly. This can be explained by suboptimal choice of the base value
that yields longer rebasing times, since the sparse representation is
overused.

\begin{figure}[t]
  \begin{center}
    \includegraphics[width=\linewidth]{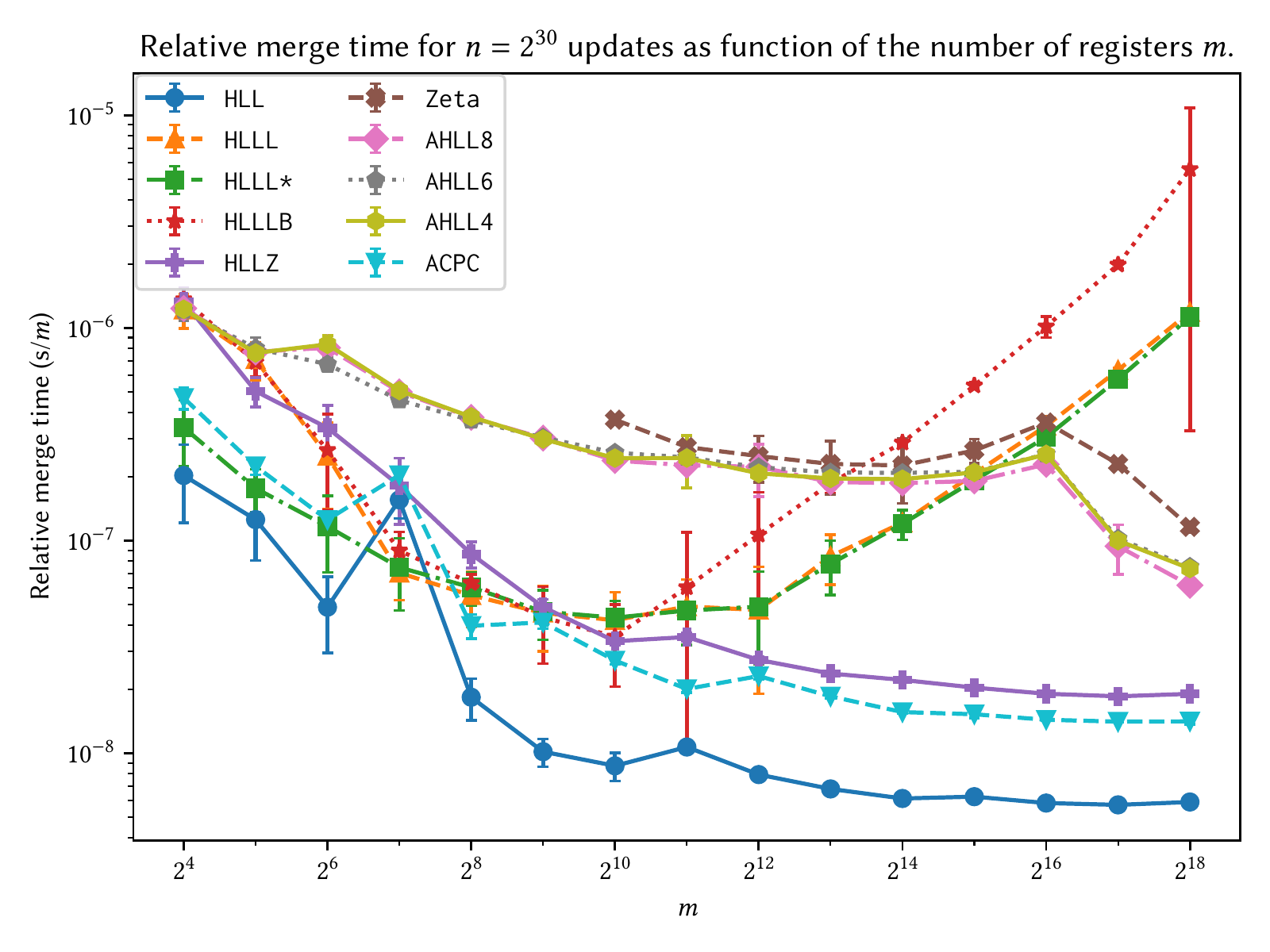}
    \caption{Relative times of merging two sketches
      constructed with $n=2^{30}$ updates.
      The times are given in seconds/$m$.
      Unsigned 64-bit integers were used as input.}
    \label{fig:merge}
  \end{center}
\end{figure}

Figure~\ref{fig:relativesketchsize} shows the relative sketch size as
bits / $m$ after $n=2^{30}$ distinct element updates. In particular,
this illustrates the efficiency of the heuristics we use for
\texttt{HLLL*} since there is no discernible difference in the sketch
size between \texttt{HLLL} and \texttt{HLLL*} when $m$ is not
minuscule. This figure also shows that \texttt{HLLLB} produces clearly
larger sketches than \texttt{HLLL*}. Although the sketch sizes
given for the Apache DataSketches algorithms are upper bounds, the
bounds appear to be quite tight as $m$ grows. In concrete numbers, for
$m=2^{18}$, \texttt{HLLL} reduces the bit size of the sketch by
37.3--37.7\%. On average over all choices of $m$, at $n=2^{30}$, we see a reduction of over 41\%.

\begin{figure}[t]
  \begin{center}
    \includegraphics[width=\linewidth]{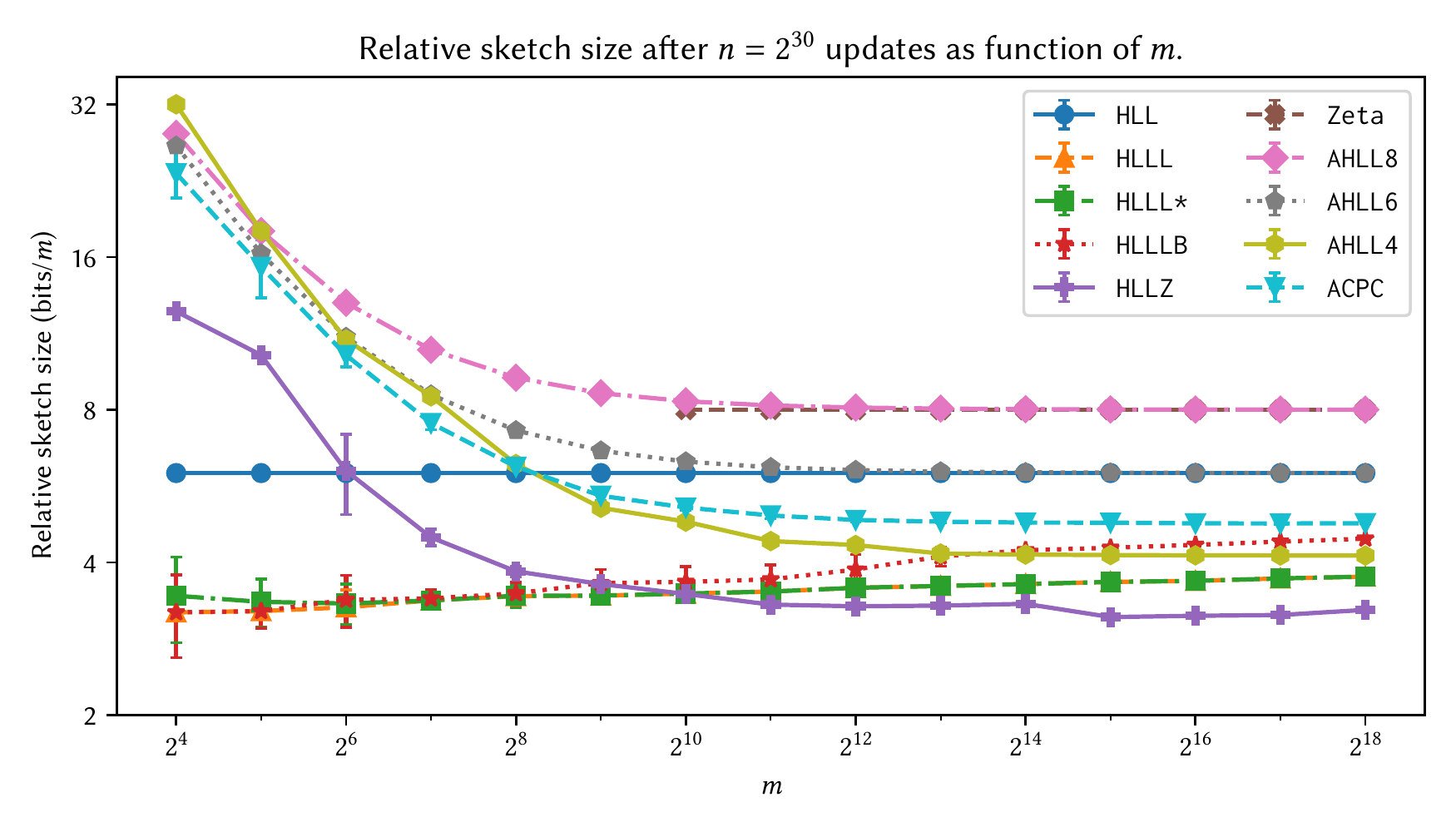}
    \caption{Relative sketch sizes (bit size / $m$) after 
      $n=2^{30}$ 64-bit integer updates as a function of $m$ in seconds.}
    \label{fig:relativesketchsize}
  \end{center}
\end{figure}

Finally, to quantify the tradeoffs among the different implementations, 
Figure~\ref{fig:pareto} plots the sketch size vs. update time
at $n=2^{30}$ and $m=2^{15}$ with unsigned
64-bit integers as input. Algorithms to the left and towards the
bottom would be preferred; this shows that \texttt{HLLL*} performs
quite well in this parameter regime. In fact, this is a sweet spot for
\texttt{HLLL*} where the sketch size is very minuscule while updates are
still within the amortized constant behavior zone. Entropy compression
is required for achieving smaller sketch size, at the expense of
considerably higher update times.

\begin{figure}[t]
  \begin{center}
    \includegraphics[width=\linewidth]{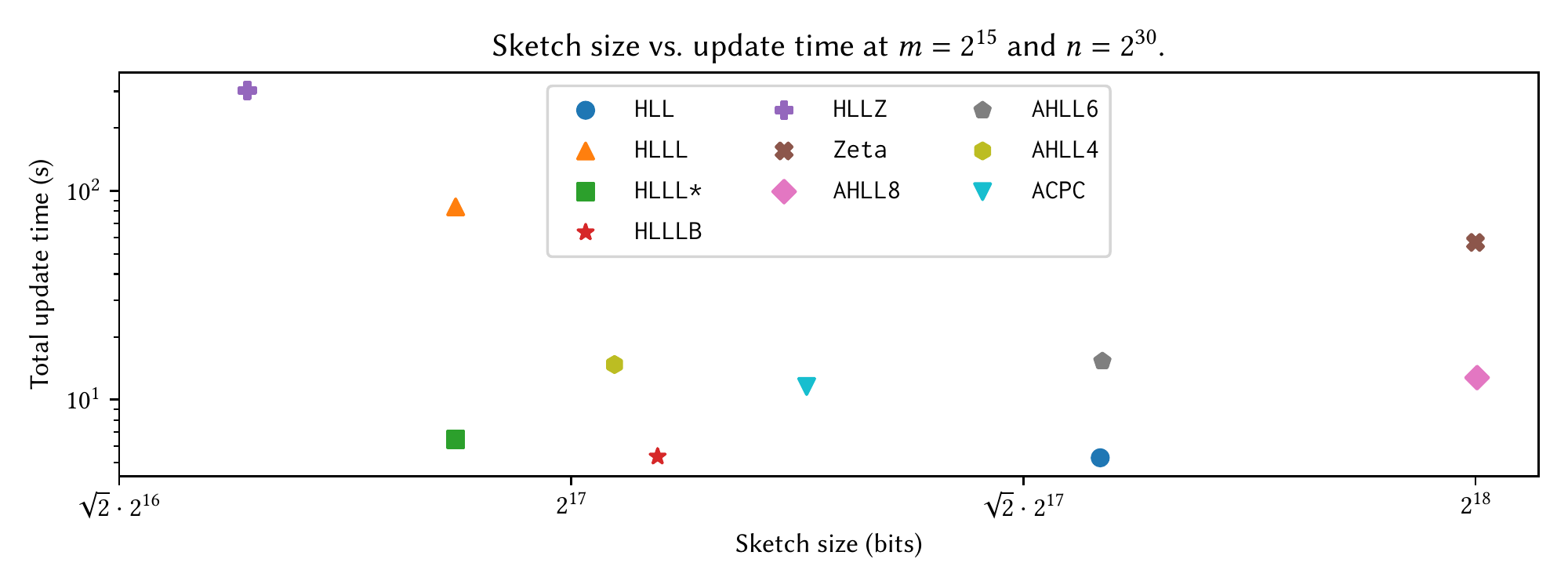}
    \caption{Update time / sketch size tradeoff plotted at $n=2^{30}$
      and $m=2^{15}$. Unsigned 64-bit integers were used as input.}
    \label{fig:pareto}
  \end{center}
\end{figure}

Figure~\ref{fig:pareto2} provides a multi-datapoint view into the
tradeoff by fixing $n=2^{30}$ and plotting the relative sketch size
vs. total update size at $m=2^{10},\ldots,2^{18}$. In general, if a
datapoint (for fixed $m$, although
this cannot be seen in the figure) lies to the left and to the bottom
of another datapoint, it should always be preferred. This shows that
there is indeed a regime where \texttt{HLLL*} is preferred when $m$ is
sufficiently small in relation to $n$, except when sketch minimality
is desired, which would favor entropy compression. \texttt{HLLL*} is at a
disadvantage when constant time updates are required at
very high accuracy, but at a relatively low $n$.

\begin{figure}[t]
  \begin{center}
    \includegraphics[width=\linewidth]{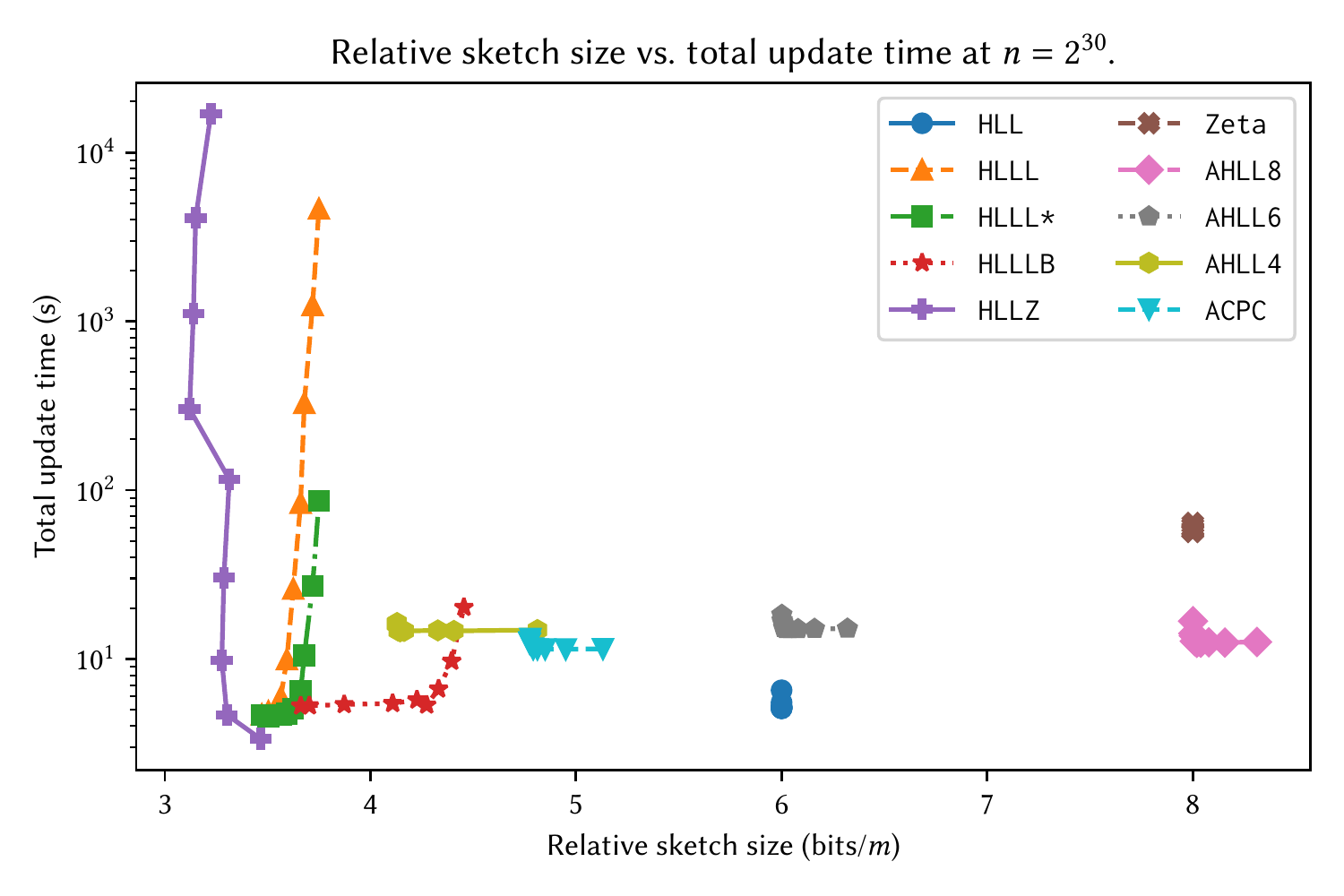}
    \caption{Update time / relative sketch size tradeoff at $n=2^{30}$ and
      $m=2^{10},\ldots,2^{18}$ with unsigned 64-bit integers as input.}
    \label{fig:pareto2}
  \end{center}
\end{figure}

\begin{acks}
  We thank Martin Aumüller for constructive discussions.
  This work was supported by VILLUM
  Foundation grant 16582.
\end{acks}

\bibliographystyle{ACM-Reference-Format}
\bibliography{refs.bib}


\begin{thebibliography}{21}


\ifx \showCODEN    \undefined \def \showCODEN     #1{\unskip}     \fi
\ifx \showDOI      \undefined \def \showDOI       #1{#1}\fi
\ifx \showISBNx    \undefined \def \showISBNx     #1{\unskip}     \fi
\ifx \showISBNxiii \undefined \def \showISBNxiii  #1{\unskip}     \fi
\ifx \showISSN     \undefined \def \showISSN      #1{\unskip}     \fi
\ifx \showLCCN     \undefined \def \showLCCN      #1{\unskip}     \fi
\ifx \shownote     \undefined \def \shownote      #1{#1}          \fi
\ifx \showarticletitle \undefined \def \showarticletitle #1{#1}   \fi
\ifx \showURL      \undefined \def \showURL       {\relax}        \fi
\providecommand\bibfield[2]{#2}
\providecommand\bibinfo[2]{#2}
\providecommand\natexlab[1]{#1}
\providecommand\showeprint[2][]{arXiv:#2}

\bibitem[\protect\citeauthoryear{Backstrom, Boldi, Rosa, Ugander, and
  Vigna}{Backstrom et~al\mbox{.}}{2012}]%
        {BackstromBRUV:2012}
\bibfield{author}{\bibinfo{person}{Lars Backstrom}, \bibinfo{person}{Paolo
  Boldi}, \bibinfo{person}{Marco Rosa}, \bibinfo{person}{Johan Ugander}, {and}
  \bibinfo{person}{Sebastiano Vigna}.} \bibinfo{year}{2012}\natexlab{}.
\newblock \showarticletitle{Four degrees of separation}. In
  \bibinfo{booktitle}{\emph{Proc. Web Science 2012}}.
  \bibinfo{publisher}{{ACM}}, \bibinfo{pages}{33--42}.
\newblock


\bibitem[\protect\citeauthoryear{Bar{-}Yossef, Jayram, Kumar, Sivakumar, and
  Trevisan}{Bar{-}Yossef et~al\mbox{.}}{2002}]%
        {Bar-YossefJKST:2002}
\bibfield{author}{\bibinfo{person}{Ziv Bar{-}Yossef}, \bibinfo{person}{T.~S.
  Jayram}, \bibinfo{person}{Ravi Kumar}, \bibinfo{person}{D. Sivakumar}, {and}
  \bibinfo{person}{Luca Trevisan}.} \bibinfo{year}{2002}\natexlab{}.
\newblock \showarticletitle{Counting Distinct Elements in a Data Stream}. In
  \bibinfo{booktitle}{\emph{Proc. {RANDOM}}}. \bibinfo{publisher}{Springer},
  \bibinfo{pages}{1--10}.
\newblock


\bibitem[\protect\citeauthoryear{Blandford and Blelloch}{Blandford and
  Blelloch}{2008}]%
        {BlandfordB:2008}
\bibfield{author}{\bibinfo{person}{Daniel~K. Blandford} {and}
  \bibinfo{person}{Guy~E. Blelloch}.} \bibinfo{year}{2008}\natexlab{}.
\newblock \showarticletitle{Compact dictionaries for variable-length keys and
  data with applications}.
\newblock \bibinfo{journal}{\emph{{ACM} Trans. Algorithms}}
  \bibinfo{volume}{4}, \bibinfo{number}{2} (\bibinfo{year}{2008}),
  \bibinfo{pages}{17:1--17:25}.
\newblock


\bibitem[\protect\citeauthoryear{B{\l{}}asiok}{B{\l{}}asiok}{2020}]%
        {Blasiok:2020}
\bibfield{author}{\bibinfo{person}{Jaros{\l{}}aw B{\l{}}asiok}.}
  \bibinfo{year}{2020}\natexlab{}.
\newblock \showarticletitle{Optimal Streaming and Tracking Distinct Elements
  with High Probability}.
\newblock \bibinfo{journal}{\emph{{ACM} Trans. Algorithms}}
  \bibinfo{volume}{16}, \bibinfo{number}{1} (\bibinfo{year}{2020}),
  \bibinfo{pages}{3:1--3:28}.
\newblock


\bibitem[\protect\citeauthoryear{Breitwieser, Baker, and Salzberg}{Breitwieser
  et~al\mbox{.}}{2018}]%
        {BreitwieserBS:2018}
\bibfield{author}{\bibinfo{person}{Florian~P. Breitwieser},
  \bibinfo{person}{Daniel~N. Baker}, {and} \bibinfo{person}{Steven~L.
  Salzberg}.} \bibinfo{year}{2018}\natexlab{}.
\newblock \showarticletitle{{Kraken\-Uniq:} confident and fast metagenomics
  classification using unique k-mer counts}.
\newblock \bibinfo{journal}{\emph{Genome Biology}} \bibinfo{volume}{19},
  \bibinfo{number}{198} (\bibinfo{year}{2018}).
\newblock


\bibitem[\protect\citeauthoryear{Dubhashi and Panconesi}{Dubhashi and
  Panconesi}{2009}]%
        {DubhashiP:2009}
\bibfield{author}{\bibinfo{person}{Devdatt~P. Dubhashi} {and}
  \bibinfo{person}{Alessandro Panconesi}.} \bibinfo{year}{2009}\natexlab{}.
\newblock \bibinfo{booktitle}{\emph{Concentration of Measure for the Analysis
  of Randomized Algorithms}}.
\newblock \bibinfo{publisher}{Cambridge University Press},
  \bibinfo{address}{New York, NY, USA}.
\newblock


\bibitem[\protect\citeauthoryear{Durand and Flajolet}{Durand and
  Flajolet}{2003}]%
        {DurandF:2003}
\bibfield{author}{\bibinfo{person}{Marianne Durand} {and}
  \bibinfo{person}{Philippe Flajolet}.} \bibinfo{year}{2003}\natexlab{}.
\newblock \showarticletitle{Loglog Counting of Large Cardinalities}. In
  \bibinfo{booktitle}{\emph{Proc. {ESA} 2003}}. \bibinfo{publisher}{Springer},
  \bibinfo{pages}{605--617}.
\newblock


\bibitem[\protect\citeauthoryear{Flajolet, Fusy, Gandouet, and
  Meunier}{Flajolet et~al\mbox{.}}{2007}]%
        {FlajoletFGM:2007}
\bibfield{author}{\bibinfo{person}{Philippe Flajolet},
  \bibinfo{person}{{\'E}ric Fusy}, \bibinfo{person}{Olivier Gandouet}, {and}
  \bibinfo{person}{Fr{\'e}d{\'e}ric Meunier}.} \bibinfo{year}{2007}\natexlab{}.
\newblock \showarticletitle{{HyperLogLog: the analysis of a near-optimal
  cardinality estimation algorithm}}. In \bibinfo{booktitle}{\emph{{AofA:
  Analysis of Algorithms}}} \emph{(\bibinfo{series}{DMTCS Proceedings})}.
  \bibinfo{pages}{137--156}.
\newblock


\bibitem[\protect\citeauthoryear{Flajolet and Martin}{Flajolet and
  Martin}{1985}]%
        {FlajoletM:1985}
\bibfield{author}{\bibinfo{person}{Philippe Flajolet} {and}
  \bibinfo{person}{G.~Nigel Martin}.} \bibinfo{year}{1985}\natexlab{}.
\newblock \showarticletitle{Probabilistic Counting Algorithms for Data Base
  Applications}.
\newblock \bibinfo{journal}{\emph{J. Comput. Syst. Sci.}} \bibinfo{volume}{31},
  \bibinfo{number}{2} (\bibinfo{year}{1985}), \bibinfo{pages}{182--209}.
\newblock


\bibitem[\protect\citeauthoryear{Giroire}{Giroire}{2009}]%
        {Giroire:2009}
\bibfield{author}{\bibinfo{person}{Fr{\'{e}}d{\'{e}}ric Giroire}.}
  \bibinfo{year}{2009}\natexlab{}.
\newblock \showarticletitle{Order statistics and estimating cardinalities of
  massive data sets}.
\newblock \bibinfo{journal}{\emph{Discret. Appl. Math.}} \bibinfo{volume}{157},
  \bibinfo{number}{2} (\bibinfo{year}{2009}), \bibinfo{pages}{406--427}.
\newblock


\bibitem[\protect\citeauthoryear{Halevy, Korn, Noy, Olston, Polyzotis, Roy, and
  Whang}{Halevy et~al\mbox{.}}{2016}]%
        {HalevyKNOPRW:2016}
\bibfield{author}{\bibinfo{person}{Alon~Y. Halevy}, \bibinfo{person}{Flip
  Korn}, \bibinfo{person}{Natalya~Fridman Noy}, \bibinfo{person}{Christopher
  Olston}, \bibinfo{person}{Neoklis Polyzotis}, \bibinfo{person}{Sudip Roy},
  {and} \bibinfo{person}{Steven~Euijong Whang}.}
  \bibinfo{year}{2016}\natexlab{}.
\newblock \showarticletitle{Goods: Organizing Google's Datasets}. In
  \bibinfo{booktitle}{\emph{Proc. {SIGMOD} 2016}}. \bibinfo{publisher}{{ACM}},
  \bibinfo{pages}{795--806}.
\newblock


\bibitem[\protect\citeauthoryear{Heule, Nunkesser, and Hall}{Heule
  et~al\mbox{.}}{2013}]%
        {HeuleNH:2013}
\bibfield{author}{\bibinfo{person}{Stefan Heule}, \bibinfo{person}{Marc
  Nunkesser}, {and} \bibinfo{person}{Alexander Hall}.}
  \bibinfo{year}{2013}\natexlab{}.
\newblock \showarticletitle{HyperLogLog in practice: algorithmic engineering of
  a state of the art cardinality estimation algorithm}. In
  \bibinfo{booktitle}{\emph{Proc. {EDBT} '13}}. \bibinfo{publisher}{{ACM}},
  \bibinfo{pages}{683--692}.
\newblock


\bibitem[\protect\citeauthoryear{Kane, Nelson, and Woodruff}{Kane
  et~al\mbox{.}}{2010}]%
        {KaneNW:2010}
\bibfield{author}{\bibinfo{person}{Daniel~M. Kane}, \bibinfo{person}{Jelani
  Nelson}, {and} \bibinfo{person}{David~P. Woodruff}.}
  \bibinfo{year}{2010}\natexlab{}.
\newblock \showarticletitle{An optimal algorithm for the distinct elements
  problem}. In \bibinfo{booktitle}{\emph{Proc. {PODS} 2010}}.
  \bibinfo{publisher}{{ACM}}, \bibinfo{pages}{41--52}.
\newblock


\bibitem[\protect\citeauthoryear{Knuth}{Knuth}{1998}]%
        {Knuth:1998}
\bibfield{author}{\bibinfo{person}{Donald~E. Knuth}.}
  \bibinfo{year}{1998}\natexlab{}.
\newblock \bibinfo{booktitle}{\emph{The Art of Computer Programming. Volume 3.
  Sorting and Searching}}.
\newblock \bibinfo{publisher}{Addison-Wesley}, \bibinfo{address}{Boston, MA,
  USA}.
\newblock
\showISBNx{978-0-201-89685-5}


\bibitem[\protect\citeauthoryear{Lang}{Lang}{2017}]%
        {Lang:2017}
\bibfield{author}{\bibinfo{person}{Kevin~J. Lang}.}
  \bibinfo{year}{2017}\natexlab{}.
\newblock \showarticletitle{Back to the Future: an Even More Nearly Optimal
  Cardinality Estimation Algorithm}.
\newblock \bibinfo{journal}{\emph{CoRR}}  \bibinfo{volume}{abs/1708.06839}
  (\bibinfo{year}{2017}).
\newblock
\showeprint[arXiv]{1708.06839}


\bibitem[\protect\citeauthoryear{Liu, Chen, and Guan}{Liu
  et~al\mbox{.}}{2016}]%
        {LiuCG:2016}
\bibfield{author}{\bibinfo{person}{Yang Liu}, \bibinfo{person}{Wenji Chen},
  {and} \bibinfo{person}{Yong Guan}.} \bibinfo{year}{2016}\natexlab{}.
\newblock \showarticletitle{Identifying High-Cardinality Hosts from
  Network-Wide Traffic Measurements}.
\newblock \bibinfo{journal}{\emph{{IEEE} Trans. Dependable Secur. Comput.}}
  \bibinfo{volume}{13}, \bibinfo{number}{5} (\bibinfo{year}{2016}),
  \bibinfo{pages}{547--558}.
\newblock


\bibitem[\protect\citeauthoryear{Pagh}{Pagh}{2001}]%
        {Pagh:2001}
\bibfield{author}{\bibinfo{person}{Rasmus Pagh}.}
  \bibinfo{year}{2001}\natexlab{}.
\newblock \showarticletitle{Low Redundancy in Static Dictionaries with Constant
  Query Time}.
\newblock \bibinfo{journal}{\emph{{SIAM} J. Comput.}} \bibinfo{volume}{31},
  \bibinfo{number}{2} (\bibinfo{year}{2001}), \bibinfo{pages}{353--363}.
\newblock


\bibitem[\protect\citeauthoryear{Paterson and Raynal}{Paterson and
  Raynal}{2021}]%
        {PatersonR:2021}
\bibfield{author}{\bibinfo{person}{Kenneth~G. Paterson} {and}
  \bibinfo{person}{Mathilde Raynal}.} \bibinfo{year}{2021}\natexlab{}.
\newblock \showarticletitle{HyperLogLog: Exponentially Bad in Adversarial
  Settings}.
\newblock \bibinfo{journal}{\emph{{IACR} Cryptol. ePrint Arch.}}
  (\bibinfo{year}{2021}), \bibinfo{pages}{1139}.
\newblock


\bibitem[\protect\citeauthoryear{Pettie, Wang, and Yin}{Pettie
  et~al\mbox{.}}{2021}]%
        {PettieWY:2021}
\bibfield{author}{\bibinfo{person}{Seth Pettie}, \bibinfo{person}{Dingyu Wang},
  {and} \bibinfo{person}{Longhui Yin}.} \bibinfo{year}{2021}\natexlab{}.
\newblock \showarticletitle{Non-Mergeable Sketching for Cardinality
  Estimation}. In \bibinfo{booktitle}{\emph{Proc. {ICALP} 2021}}
  \emph{(\bibinfo{series}{LIPIcs})}, Vol.~\bibinfo{volume}{198}.
  \bibinfo{publisher}{Schloss Dagstuhl - Leibniz-Zentrum f{\"{u}}r Informatik},
  \bibinfo{pages}{104:1--104:20}.
\newblock


\bibitem[\protect\citeauthoryear{Whang, Zanden, and Taylor}{Whang
  et~al\mbox{.}}{1990}]%
        {WhangVT:1990}
\bibfield{author}{\bibinfo{person}{Kyu{-}Young Whang}, \bibinfo{person}{Brad
  T.~Vander Zanden}, {and} \bibinfo{person}{Howard~M. Taylor}.}
  \bibinfo{year}{1990}\natexlab{}.
\newblock \showarticletitle{A Linear-Time Probabilistic Counting Algorithm for
  Database Applications}.
\newblock \bibinfo{journal}{\emph{{ACM} Trans. Database Syst.}}
  \bibinfo{volume}{15}, \bibinfo{number}{2} (\bibinfo{year}{1990}),
  \bibinfo{pages}{208--229}.
\newblock


\bibitem[\protect\citeauthoryear{Xiao, Chen, Zhou, and Luo}{Xiao
  et~al\mbox{.}}{2020}]%
        {XiaoCZL:2020}
\bibfield{author}{\bibinfo{person}{Qingjun Xiao}, \bibinfo{person}{Shigang
  Chen}, \bibinfo{person}{You Zhou}, {and} \bibinfo{person}{Junzhou Luo}.}
  \bibinfo{year}{2020}\natexlab{}.
\newblock \showarticletitle{Estimating Cardinality for Arbitrarily Large Data
  Stream With Improved Memory Efficiency}.
\newblock \bibinfo{journal}{\emph{{IEEE/ACM} Trans. Netw.}}
  \bibinfo{volume}{28}, \bibinfo{number}{2} (\bibinfo{year}{2020}),
  \bibinfo{pages}{433--446}.
\newblock


\end{thebibliography}

\appendix

\end{document}